\documentclass{article}
\usepackage{etex}
\usepackage{amsmath}
\usepackage{tikz}
\usepackage{tkz-graph}
\usepackage{bussproofs}
\usepackage{graphicx}

\usetikzlibrary{arrows,automata}
\usetikzlibrary{arrows,shapes,positioning}
\usetikzlibrary{tikzmark,decorations.pathreplacing,calc}
\newtheorem{theorem}{Theorem}

\newtheorem{corollary}[theorem]{Corollary}

\newtheorem{definition}[theorem]{Definition}

\newtheorem{lemma}[theorem]{Lemma}

\newtheorem{remark}[theorem]{Remark}

\newtheorem{partial solution}[theorem]{Partial Solution}

\newenvironment{proof}[1][Proof]{\textbf{#1.} }{\ \rule{0.5em}{0.5em}}

\input{tcilatex}

\begin{document}

\begin{center}
{\Large On Proof Theory in Computational Complexity\smallskip }

L. Gordeev, E. H. Haeusler, LC2021, July 24, 2021
\end{center}

\section{Background}

\begin{itemize}
\item  Earlier in 20th century logic proofs (deductions, derivations) were
understood like this: \emph{Formula $\varphi $ is derivable from axioms $%
\mathcal{A}$ (: $\mathcal{A}\vdash \varphi $) iff $\left( \exists \varphi
_{0},\cdots ,\varphi _{n}\right) \varphi =\varphi _{n}\,$\emph{and}}%
\thinspace \emph{$\left( \forall k\leq n\right) \varphi _{k}\in \mathcal{A\,}
$or\thinspace $\left( \exists i,j<k\right) \varphi _{j}=\varphi
_{i}\rightarrow \varphi _{k}$, i.e. $\varphi _{k\ }$follows from $\varphi
_{i}$ and $\varphi _{j}$ by the rule ``modus ponens'' (detachment)}. Other
rules of inference could be included analogously.

\item  Proofs in the algebraic logic (boolean and relation algebras)
were/are treated analogously with regard to the transitivity of ``=",
instead of \emph{modus ponens}.

\item  These definitions reflected traditional written \emph{linear}
presentation of mathematical proofs of new theorems via axioms, known
theorems, suitable lemmas, etc.

\item  Corresponding proof systems in mathematical logic are usually
referred to as Frege-Hilbert-Bernays-Tarski style calculi.
\end{itemize}

\begin{itemize}
\item  Later came graph-theoretic interpretations leading to genuine \emph{%
structural} proof theory.

\item  Corresponding basic proof systems are referred to as \emph{natural
deduction} (ND) and \emph{sequent calculus} (SC) -- both initiated by
Gentzen and further developed by Prawitz resp. Sch\"{u}tte, et al. These
proofs are usually presented in tree-like form, where branching points are
determined by the conclusions of the rules involved. Moreover

\begin{enumerate}
\item  ND derivations contain single formulas, whereas SC ones contain
finite collections thereof (called sequents).

\item  ND proofs have no axioms. However all assumptions shall be \emph{%
discharged} according to special conditions on the threads.
\end{enumerate}

\item  Both ND and SC allow \emph{normalizations} (mutually different)
making proofs more transparent and suitable for analysis.
\end{itemize}

\section{ Proof size}

\begin{itemize}
\item  Linear proofs admit tree-like interpretation, and v.v. Different
nodes in tree-like proofs might correspond to identical formulas
(``references'') $\varphi _{i}$, $\varphi _{j}$ occurring in linear proofs $%
\mathcal{A}\vdash \varphi $ (sequent case is analogous). So passing to
tree-like proofs might essentially increase the size of linear inputs. The
opposite direction is called \emph{proof compression}. Actually we compress
tree-like proofs into \emph{dag-like} proofs (\emph{dag} = directed acyclic
graph) by merging different nodes labeled with identical formulas
(sequents). Moreover
\end{itemize}

\begin{enumerate}
\item  Proof compression in SC is easy. However, we can't really control the
size of resulting dag-like proofs, as there can be too many different
(sub)sequents occurring in (even normal) proofs of given ``small''
conclusions.

\item  In contrast, ND proofs contain single formulas thus being more
appropriate for the size control. However, proof compression in ND is more
involved.
\end{enumerate}

\section{Minimal logic}

\begin{itemize}
\item  In this work we consider basic ND of minimal purely implicational
logic, \textsc{NM}$_{\rightarrow }$, having two standard rules of inferences 
\begin{equation*}
\fbox{$\left( \rightarrow I\right) :\dfrac{\QDATOP{\QDATOP{\left[ \alpha %
\right] }{\vdots }}{\beta }}{\alpha \rightarrow \beta }$}\quad \fbox{$\left(
\rightarrow E\right) :\dfrac{\alpha \quad \alpha \rightarrow \beta }{\beta \ 
}$ (= \emph{modus ponens})}
\end{equation*}
and auxiliary repetition rule 
\begin{equation*}
\fbox{$\left( R\right) :\dfrac{\alpha }{\alpha \ }$}
\end{equation*}
where $\left[ \alpha \right] $ in $\left( \rightarrow I\right) $\ indicates
that all $\alpha $-leaves occurring above $\beta $-node exposed are
considered \emph{discharged} assumptions.
\end{itemize}

\begin{definition}[\emph{minimal validity}]
A given (whether tree- or dag-like) \textsc{NM}$_{\rightarrow }$-deduction $%
\partial $ \emph{proves} its conclusion $\rho $ (abbr.: $\partial \vdash
\rho $) iff every maximal thread connecting the root labeled $\rho $ with a
leaf labeled $\alpha $ is \emph{closed}, i.e. it contains a $\left(
\rightarrow I\right) $ with conclusion $\alpha \rightarrow \beta $ and
discharged assumption $\alpha $, for some $\beta $. Now $\rho $ is \emph{%
valid in minimal logic} iff there exists a tree-like \textsc{NM}$%
_{\rightarrow }$-deduction $\partial $ that proves $\rho \,$; such $\partial 
$ is called a \emph{proof} of $\rho $.
\end{definition}

\begin{remark}
Tree-like constraint in the definition of validity is inessential, as any
dag-like $\partial $ can be unfolded into a tree-like $\partial ^{\prime }$
by thread-preserving top-down recursion. Moreover, ``$\,\partial $ proves $%
\rho $'' is deterministically verifiable in $\left| \partial \right| $%
-polynomial time, where $\left| \partial \right| $ denotes the weight of $%
\partial $.
\end{remark}

\begin{definition}
A given\emph{\ }\textsc{NM}$_{\rightarrow }$-deduction $\partial $ with
conclusion $\rho $ is \emph{polynomial}, resp. \emph{quasi-polynomial}, if
its weight (= total number of symbols) $\left| \partial \right| $, resp.
height $h\left( \partial \right) $ plus total weight $\phi \left( \partial
\right) $ of distinct formulas occurring in $\partial $, is polynomial in
the weight of $\rho $, $\left| \rho \right| $. Note that $\left| \partial
\right| $ of quasi-polynomial $\partial $ can be exponential in $\left| \rho
\right| $.
\end{definition}

\begin{theorem}[\emph{Main Theorem}]
Any given quasi-polynomial tree-like proof $\partial \vdash \rho $ can be
compressed into a polynomial dag-like proof $\partial ^{\ast }\vdash \rho $.
\end{theorem}

\begin{proof}
See GH1, 2 that presented desired \emph{horizontal compression} of
quasi-polynomial tree-like proofs into equivalent polynomial dag-like proofs
having mutually different formulas on every horizontal level (see also
Section 5 below).
\end{proof}

\section{Propositional complexity.}

\subsection{Case $\mathbf{NP}$ vs $\mathbf{coNP}$}

\begin{lemma}
Any normal\ tree-like \textsc{NM}$_{\rightarrow }$-proof $\partial $ of $%
\rho $ whose height $h\left( \partial \right) $ is polynomial in $\left|
\rho \right| $ is quasi-polynomial.
\end{lemma}

\begin{lemma}[GH3]
Let $P$ be the Hamiltonian graph problem and $\rho _{G}$ express that a
given graph $G$ has no Hamiltonian cycles. There exists a normal tree-like 
\textsc{NM}$_{\rightarrow }$-proof $\partial $ of $\rho _{G}$ such that $%
h\left( \partial \right) $ is polynomial in $\left| G\right| $ (and hence $%
\left| \rho _{G}\right| $), provided that $G$ is non-Hamiltonian.
\end{lemma}

Recall that the non-hamiltoniancy in question is $\mathbf{coNP}$-complete.
Hence Theorem 4 yields

\begin{corollary}[GH2, GH3]
$\mathbf{NP=coNP}$\textbf{\ }holds true.
\end{corollary}

\subsection{Case $\mathbf{NP}$ vs $\mathbf{PSPACE}$}

Recall that the minimal validity is $\mathbf{PSPACE}$-complete. Let \textsc{%
LM}$_{\rightarrow }$ be Hudelmaier's SC that is sound and complete for
minimal logic.

\begin{theorem}[Hudelmaier]
Any formula $\rho $ is valid in minimal logic iff sequent $\Rightarrow \rho $
is provable in \textsc{LM}$_{\rightarrow }$ by a quasi-polynomial tree-like
derivation.
\end{theorem}

\begin{lemma}[GH1]
For any quasi-polynomial tree-like derivation of $\Rightarrow \rho $ in 
\textsc{LM}$_{\rightarrow }$ there exists a quasi-polynomial tree-like proof 
$\partial \vdash \rho $ in \textsc{NM}$_{\rightarrow }$.
\end{lemma}

\begin{corollary}[GH2]
$\mathbf{PSPACE}\subseteq \mathbf{NP}$ and hence $\mathbf{NP=PSPACE}$ holds
true.
\end{corollary}

\begin{remark}
Using $\mathbf{PSPACE}$-completeness of quantified boolean logic V. Sopin
claimed to have obtained a partial result $\mathbf{PH=PSPACE}$.
\end{remark}

\section{ More on Main Theorem}

\begin{itemize}
\item  \textbf{First part of tree-to-dag horizontal compression}
\end{itemize}

For any tree-like \textsc{NM}$_{\rightarrow }$ proof $\partial $ of $\rho \,$
let $\partial ^{\prime }\in \,$\textsc{NM}$_{\rightarrow }$ be defined by
bottom-up recursion on $h\left( \partial \right) $ such that for any $n\leq
h\left( \partial \right) $, the $n^{th}$\ horizontal section of $\partial
^{\prime }$ is obtained by merging all nodes with identical formulas
occurring in the $n^{th}$\ horizontal section of $\partial $. The inferences
in $\partial ^{\prime }$ are naturally inherited by the ones in $\partial $.
Obviously $\partial ^{\prime }$ is a dag-like (not necessarily tree-like
anymore) deduction with conclusion $\rho $. Moreover $\partial ^{\prime }$
is polynomial as $\left| \partial ^{\prime }\right| \leq h\left( \partial
\right) \times \phi \left( \partial \right) $. However, $\partial ^{\prime }$
need not preserve the local correctness with respect to basic inferences $%
\left( \rightarrow I\right) $, $\left( \rightarrow E\right) $, $\left(
R\right) $. For example, a compressed multipremise configuration 
\begin{equation*}
\fbox{$\ \left( \rightarrow I,E\right) :\dfrac{\beta \quad \quad \gamma
\quad \quad \gamma \rightarrow \left( \alpha \rightarrow \beta \right) }{%
\alpha \rightarrow \beta }$}
\end{equation*}
that is obtained by merging identical conclusions $\alpha \rightarrow \beta $
of 
\begin{equation*}
\fbox{$\ \left( \rightarrow I\right) :\dfrac{\beta }{\alpha \rightarrow
\beta }$}\quad \text{and\quad }\fbox{$\ \left( \rightarrow E\right) :\dfrac{%
\gamma \quad \quad \gamma \rightarrow \left( \alpha \rightarrow \beta
\right) }{\alpha \rightarrow \beta }$}
\end{equation*}
is not a legitimate inference in \textsc{NM}$_{\rightarrow }$.

To overcome this trouble we upgrade $\partial ^{\prime }$ to a modified
deduction $\partial ^{\flat }$ that separates such multiple premises using
instances of the separation rule $\left( S\right) $ 
\begin{equation*}
\fbox{$\left( S\right) :\dfrac{\overset{n\ times}{\overbrace{\alpha \quad
\cdots \quad \alpha }}}{\alpha \ }\ $($n$ arbitrary) }
\end{equation*}
that is understood disjunctively: ``\emph{if at least one premise is proved
then so is the conclusion}'' (in contrast to ordinary inferences: ``\emph{if
all premises are proved then so are the conclusions}'').

For example, $\left( \rightarrow I,E\right) $ as above should be replaced by
this modified configuration in \textsc{NM}$_{\rightarrow }^{\flat }$ = 
\textsc{NM}$_{\rightarrow }$ + $\left( S\right) $ 
\begin{equation*}
\fbox{$\left( S\right) :\dfrac{\ \left( \rightarrow I\right) :\ \dfrac{\beta 
}{\alpha \rightarrow \beta \ }\quad \left( \rightarrow E\right) \ :\dfrac{%
\gamma \quad \quad \gamma \rightarrow \left( \alpha \rightarrow \beta
\right) }{\alpha \rightarrow \beta \ }}{\ \alpha \rightarrow \beta \ }$}%
\text{.}
\end{equation*}
Such $\partial ^{\flat }$ is a locally correct dag-like deduction in \textsc{%
NM}$_{\rightarrow }^{\flat }$ with conclusion $\rho $. Moreover $\partial
^{\flat }$ is polynomial, since its every $\left( S\right) $-free
subdeduction at most doubles the weight of $\partial ^{\prime }$. However,
we can't claim that $\partial ^{\flat }$ proves $\rho $ because arbitrary
maximal dag-like threads in $\partial ^{\flat }$ can arise by concatenating
different segments of different threads in $\partial $, which can destroy
the required closure condition.

To solve this problem we observe that $\partial ^{\flat }$ satisfies certain
conditions of \emph{coherency} with respect to the set of threads, and
continue our compression as follows.

\begin{itemize}
\item  \textbf{Second part of tree-to-dag horizontal compression}
\end{itemize}

Here we prove \emph{weak $\left( S\right) $-elimination theorem} showing
that any coherent deduction $\partial ^{\flat }$ is further compressible
into a desired $\left( S\right) $-free subdeduction $\partial ^{\ast }$.
This part of compression (also called \emph{cleansing}) is defined by
nondeterministic bottom-up recursion on $h\left( \partial \right) $ while
using as oracle the whole (possibly exponential) set of maximal threads. 
\textbf{This completes proof of Main Theorem. \footnote{{\footnotesize see
GH1, GH2 for details}}}

\section{ More on dag-like provability}

Formal verification of the assertion $\partial \vdash \rho $ is simple --
whether for tree-like or, generally, dag-like $\partial $. Every node $x\in
\partial $ is assigned, by top-down recursion, a set of assumptions $A\left(
x\right) $ such that:

\begin{enumerate}
\item  $A\left( x\right) :=\left\{ \alpha \right\} $ if $x$ is a leaf
labeled $\alpha $,

\item  $A\left( x\right) :=A\left( y\right) $ if $x$ is the conclusion of $%
\left( R\right) $ with premise $y$,

\item  $A\left( x\right) :=A\left( y\right) \setminus \left\{ \alpha
\right\} $ if $x$ is the conclusion of $\left( \rightarrow I\right) $ with
label $\alpha \rightarrow \beta $\ and premise $y$,

\item  $A\left( x\right) :=A\left( y\right) \cup A\left( z\right) $ if $x$
is the conclusion of $\left( \rightarrow E\right) $ with premises $y,$ $z$.
\end{enumerate}

\begin{theorem}
$\partial \vdash \rho \Leftrightarrow A\left( r\right) =\emptyset $ holds
with respect to standard set-theoretic interpretations of ``$\,\cup $'' and
``$\,\setminus $'', where $r$ is the root of $\partial $ with formula-label $%
\rho $. Moreover, problem $A\left( r\right) \overset{?}{=}\emptyset $ is
solvable by a deterministic TM in $\left| \partial \right| $-polynomial time.
\end{theorem}

\section{References}

\quad\ GH1: L. Gordeev, E. H. Haeusler, \emph{Proof Compression and NP
Versus PSPACE}, Studia Logica (107) (1): 55--83 (2019)

GH2: L. Gordeev, E. H. Haeusler, \emph{Proof Compression and NP Versus
PSPACE II}, Bulletin of the Section of Logic (49) (3): 213--230 (2020),

http://dx.doi.org/10.18788/0138-0680.2020.16

GH3: L. Gordeev, E. H. Haeusler, \emph{Proof Compression and NP Versus
PSPACE II: Addendum}, Bulletin of the Section of Logic (51), 9 pp. (2022)

http://dx.doi.org/10.18788/0138-0680.2022.01

Hudelmaier, \emph{An }$O\left( n\log n\right) $\emph{-space decision
procedure for intuitionistic propositional logic}, J. Logic Computat. (3):
1--13 (1993)

D. Prawitz, \textbf{Natural deduction: a proof-theoretical study}. Almqvist
\& Wiksell, 1965

V. Sopin, \emph{PH=PSPACE}, https://arxiv.org/abs/1411.0628

\end{document}